\documentclass{article}
\usepackage{natbib}

\usepackage{amsmath,amssymb,amsthm}
\newtheorem{theorem}{Theorem}
\newtheorem{corollary}[theorem]{Corollary}
\newtheorem{definition}[theorem]{Definition}
\newtheorem{lemma}[theorem]{Lemma}

\begin{document}
\title{Reichenbachian common cause systems of size 3 \\ in general probability theories}
\author{Yuichiro Kitajima}
\date{}
\maketitle

\begin{abstract}
Reichenbach defined a common cause which explains a correlation between two events if either one does not cause the other. Its intuitive idea is that the statistical ensemble can be divided into two disjoint parts so that the correlation disappears in both of the resulting subensembles if there is no causal connection between these correlated events. These subensembles can be regarded as common causes. Hofer-Szab\'{o} and R\'{e}dei (2004) generalized a Reichenbachian common cause, and called it a Reichenbachian common cause system. In the case of Reichenbachian common cause systems the statistical ensemble is divided more than two, while it is divided into two parts in the case of Reichenbachian common causes. The number of these subensembles is called the size of this system. In the present paper, we examine Reichenbachian common cause systems in general probability theories which include classical probability theories and quantum probability theories. It is shown that there is no Reichenbachian common cause system for any correlation between two events which are not logical independent, and that a general probability theory which is represented by an atomless orthomodular lattice with a faithful $\sigma$-additive probability measure contains Reichenbachian common cause systems of size 3 for any correlation between two compatible and logical independent events. Moreover, we discuss a relation between Bell's inequality and Reichenbachian common cause systems, and point out that this violation of Bell's inequality can be compatible with a Reichenbachian common cause system although it contradicts a `common' common cause system.

\end{abstract}

\section{Introduction}

\cite{reichenbach1956direction} defined a common cause which explain correlations between two events if either one does not cause the other. Its intuitive idea is that the statistical ensemble can be divided into two disjoint parts so that the correlation disappears in both of the resulting subensembles. These subensembles are regarded as common causes. Reichenbachain common cause principle states that a correlation between two events  is either due to a direct causal link between the correlated events, or there is a Reichenbachian common cause that explains the correlation.

\cite{gyenis2004can} introduced the notion of common cause closedness, which means that there always exists a Reichenbachian common cause for any correlation between two events if either one does not cause the other. In a common cause closed probability measure, Reichenbachian common cause principle holds. When a probability measure space is not common cause closed, there are two strategies to save this principle. One strategies is to think that such a probability measure space is not rich enough to contain a common cause, and it can be embedded into a larger one which contains a common cause of the correlation. Such a larger probability space is called common cause complete \cite[Chapter 3]{hofer2013principle}.

Another strategy is to suspect that the correlation is not due to two subensembles but the result of more than two subensembles. To formulate the latter idea, \cite{hofer2004reichenbachian, hofer2006reichenbachian} generalized a Reichenbachian common cause, and called it a Reichenbachian common cause system. The statistical ensemble is divided more than two in the case of Reichenbachian common cause systems while it is divided into two parts in the case of Reichenbachian common causes. The number of these subensembles is called the size of this system. 
A corresponding notion of common cause closedness can be introduced; a probability measure space is called causally $n$-closed if it contains Reichenbachian common cause systems whose subensembles are $n$ for any correlation between two causally independent events.

Common causes closedness can be defined in quantum probability theories as well as classical probability theories \citep{gyenis2014atomicity,kitajima2015characterizing}. By a quantum probability theories are meant orthomodular lattices of projections of a noncommutative von Neumann algebra with $\sigma$-additive probability measure on this lattice. \cite{kitajima2015characterizing} characterized common cause closedness of quantum probability theories. According to this result, a quantum probability space is common cause closed if and only if it has at most one measure theoretic atom.

The notion of Reichenbachian common cause systems also can be formulated  in quantum probability theories as well as common cause closedness, and a more general notion of causal $n$-closedness of quantum probability theories can be defined naturally. Recently \citet[Section 6]{wronski2014new} examine causal $n$-closedness in an orthomodular lattice which has atoms. 

Reichenbachian common cause systems are defined for the correlation between two events if either one does not cause the other. So it is important whether a causal connection between the two events exists or not. In orthomodular lattices, no causal connection is represented by logical independence \citep{redei1995logical, redei1995logically, redei1998quantum}. In the present paper, we examine Reichenbachian common cause systems in orthomodular lattices. Concretely speaking, the following two problems are investigated:
\begin{itemize}
\item Is logical independence necessary for the existence of a Reichenbachian common cause system?
\item Is there a general probability theory which is causal $n$-closed for some natural number $n$?
\end{itemize}
It is shown that there is no Reichenbachian common cause system for any correlation between two events if they are not logical independent (Theorem \ref{logical-independence-rccs-proposition}), and that a general probability theory which is represented by an atomless orthomodular lattice with a faithful $\sigma$-additive probability measure contains Reichenbachian common cause systems of size 3 for any correlation between two causally independent events (Theorem \ref{size-3}). In other words, any correlation in this probability theory can be explained by a Reichenbachian common cause system of size 3. 

A reason why an atomless orthomodular lattice is examined is that it can be applied to algebraic quantum field theory (AQFT). AQFT predicts many states which violate Bell's inequality between two events associated with space-like separated spacetime regions in Minkowski spacetime \citep{halvorson2000generic, kitajima2013epr, landau1987violation, summers1987bell, summers1987abell, summers1987maximal, summers1988maximal}. Does the existence of a Reichenbachian common cause system contradict this violation of Bell's inequality? Bell's inequality holds if it is assumed that there is a local hidden variable. This local hidden variable can be regarded as a `common' common cause. Thus, no `common' common cause exists if Bell's inequality does not hold. Because a Reichbachian common cause is different from `common' common cause, it can be compatible with the violation of Bell's inequality \citep{redei1997reichenbach, hofer1999reichenbach, redei2002local}. After Corollary \ref{corollary_size-3}, we discuss a relation between Bell's inequality and Reichenbachian common cause systems, and point out that this violation of Bell's inequality can be compatible with a Reichenbachian common cause system although it contradicts a `common' common cause system.


\section{Reichenbachian common cause systems}
An orthocomplemented lattice $\mathcal{L}$ with lattice operations $\vee$, $\wedge$, and orthocomplementation $\perp$ is called orthomodular if, for any $A, B \in \mathcal{L}$ such that $A \leq B$,
\begin{equation}
B=A \vee (A^{\perp} \wedge B).
\end{equation}
Throughout the paper $\mathcal{L}$ denotes an orthomodular lattice. An example of orthomodular lattices $\mathcal{L}$ is a Boolean algebra, which is distributive, i.e. if for any $A,B,C \in \mathcal{L}$
\[ A \vee (B \wedge C)=(A \vee B) \wedge (A \vee C). \]
Other examples of orthomodular lattices are the lattices of projections of a von Neumann algebra. The lattices of projections of a von Neumann algebra is distributive if and only if the von Neumann algebra is commutative. Generally an orthomodular lattice is not necessarily distributive.

The two elements $A, B \in \mathcal{L}$ be called compatible if 
\begin{equation}
\label{compatibility-equality}
A=(A \wedge B) \vee (A \wedge B^{\perp}). 
\end{equation}
Equation (\ref{compatibility-equality}) holds if and only if
\begin{equation}
B=(B \wedge A) \vee (B \wedge A^{\perp}).
\end{equation}
In other words, the compatibility relation is symmetric \cite[Theorem 3.2]{kalmbach1983orthomodular}. 
If $A \leq B$, then $A$ and $B$ are compatible because $B=A \vee (A^{\perp} \wedge B)=(B \wedge A) \vee (B \wedge A^{\perp})$.

In the present paper, it is assumed that orthomodular lattices are bounded: they have a smallest and a largest element denoted by $0$ and $1$, respectively. If for every countable subset $S$ of $\mathcal{L}$, the join and the meet of all elements in $S$ exist, then $\mathcal{L}$ is called a $\sigma$-complete orthomodular lattice. An orthomodular lattice is called atomless if, for any nonzero element $A \in \mathcal{L}$, there exists $B \in \mathcal{L}$ such that $0 < B < A$. For example, the orthomodular lattices of all projections on a Hilbert space is not atomless because one-dimensional projections are atoms in this case. On the other hand, the orthomodular lattices of projections of type II or type III von Neumann algebras are atomless orthomodular lattices. It is well known that typical local algebras in algebraic quantum field theory are of type III \cite[Section V.6]{haag1996local}, that is, the projection lattice of these algebras are atomless.

Let $\mathcal{L}$ be a $\sigma$-complete orthomodular lattice. Elements $A$ and $B$ in $\mathcal{L}$ are called mutually orthogonal if $A \leq B^{\perp}$. The map $\phi: \mathcal{L} \rightarrow [0,1]$ is called a probability measure on $\mathcal{L}$ if $\phi(1)=1$ and $\phi(A \vee B)=\phi(A)+\phi(B)$ for any mutually orthogonal elements $A$ and $B$. A probability measure $\phi$ is called a $\sigma$-additive probability measure on $\mathcal{L}$ if for any countable, mutually orthogonal elements $\{ A_i | i \in \mathbb{N} \}$, 
\[ \phi(\vee_{i \in \mathbb{N}} A_i)=\sum_{i \in \mathbb{N}} \phi(A_i). \]
If $\phi(A)=0$ implies $A=0$ for any $A \in \mathcal{L}$, then $\phi$ is called faithful.

Let $\mathcal{L}$ be an orthomodular lattice, $\phi$ be a probability measure on $\mathcal{L}$. We call $(\mathcal{L}, \phi)$ a nonclassical probability space. If $\mathcal{L}$ is a Boolean algebra, $(\mathcal{L}, \phi)$ is called a classical probability space. Thus nonclassical probability measure space is a more general notion than a classical one. An example of this space is a quantum probability space. 

\begin{definition} 
Let $\mathcal{L}$ be an orthomodular lattice, and let $\phi$ be a probability measure on $\mathcal{L}$. If $A$ and $B$ in $\mathcal{L}$ are compatible and
\[ \phi(A \wedge B)>\phi(A)\phi(B), \]
then $A$ and $B$ are called correlated with respect to $\phi$.
\end{definition}

The following lemmas are needed in this paper.

\begin{lemma}
\label{inequality_lemma}
Let $\mathcal{L}$ be an orthomodular lattice, and let $\phi$ be a probability measure on $\mathcal{L}$. Let $A$, $B$, and $C$ in $\mathcal{L}$ be mutually compatible elements in $\mathcal{L}$. Then 
\[ \phi(A \wedge C)\phi(B \wedge C) \geq \phi((A \wedge B) \wedge C)\phi((A \vee B) \wedge C) \]
\end{lemma}

\begin{proof}
By \citet[p.25]{kalmbach1983orthomodular}, the sublattice generated by $A$, $B$, and $C$ is distributive. Thus
\[
\begin{split}
\phi(A \wedge C)\phi(B \wedge C) &= \phi(A \wedge C)(\phi(B \wedge C) - \phi(A \wedge B \wedge C))+\phi(A \wedge C)\phi(A \wedge B \wedge C) \\
&\geq \phi(A \wedge B \wedge C)(\phi(B \wedge C) - \phi(A \wedge B \wedge C))+\phi(A \wedge C)\phi(A \wedge B \wedge C) \\
&=\phi((A \wedge B) \wedge C)(\phi(A \wedge C)+\phi(B \wedge C)-\phi(A \wedge B \wedge C)) \\
&=\phi((A \wedge B) \wedge C)\phi((A \wedge C) \vee ((B \wedge C) \wedge (A \wedge B \wedge C)^{\perp})) \\
&=\phi((A \wedge B) \wedge C)\phi((A \vee B) \wedge C).
\end{split}
\]
\end{proof}

\begin{lemma}
\label{correlation-lemma}
Let $\mathcal{L}$ be an orthomodular lattice, let $\phi$ be a probability measure on $\mathcal{L}$, and let $A$ and $B$ be elements in $\mathcal{L}$ such that $\phi(A \wedge B) > \phi(A)\phi(B)$. Then $\phi(A \vee B) < 1$ and $\phi(A \wedge B) > 0$.
\end{lemma}

\begin{proof}
Let $A$ and $B$ be mutually compatible elements in $\mathcal{L}$ such that 
\begin{equation}
\label{lemma-1}
\phi(A \wedge B)>\phi(A)\phi(B). 
\end{equation}
Then $\phi(A \wedge B)>\phi(A)\phi(B) \geq 0$.
By Lemma \ref{inequality_lemma}
\begin{equation}
\label{lemma-2}
\phi(A)\phi(B) \geq \phi(A \wedge B)\phi(A \vee B)
\end{equation}
Equations (\ref{lemma-1}) and (\ref{lemma-2}) imply 
\begin{equation}
\phi(A \vee B) < 1.
\end{equation}
\end{proof}

\begin{lemma}
\label{denseness_lemma}
\cite[Lemma 3.6]{kitajima2008reichenbach} \cite[Proposition 3.5]{gyenis2014atomicity}
Let $\mathcal{L}$ be a $\sigma$-complete atomless orthomodular lattice, $\phi$ a $\sigma$-additive probability measure on $\mathcal{L}$ and $A$ an element in $\mathcal{L}$ such that $0 < \phi(A)$. For any real number $x$ such that $0 < x < \phi(A)$, there is an element $X \in \mathcal{L}$ such that $\phi(X)=x$ and $0<X<A$.
\end{lemma}

Reichenbachian common causes in a general probability theory are defined as follows.

\begin{definition}
\cite[Definition 6.1]{hofer2013principle}
Let $\mathcal{L}$ be an orthomodular lattice, and let $\phi$ be a probability measure on $\mathcal{L}$. 
If $A$ and $B$ in $\mathcal{L}$ are correlated, then $C \in \mathcal{L}$ is called a common cause for the correlation if $C$ is compatible with both $A$ and $B$, and the following conditions hold:
\begin{equation}
\frac{\phi(A \wedge C)}{\phi(C)}\frac{\phi(B \wedge C)}{\phi(C)}=\frac{\phi(A \wedge B \wedge C)}{\phi(C)}
\end{equation}
\begin{equation}
\frac{\phi(A \wedge C^{\perp})}{\phi(C^{\perp})}\frac{\phi(B \wedge C^{\perp})}{\phi(C^{\perp})}=\frac{\phi(A \wedge B \wedge C^{\perp})}{\phi(C^{\perp})}
\end{equation}
\begin{equation}
\frac{\phi(A \wedge C)}{\phi(C)} > \frac{\phi(A \wedge C^{\perp})}{\phi(C^{\perp})}
\end{equation}
\begin{equation}
\frac{\phi(B \wedge C)}{\phi(C)} > \frac{\phi(B \wedge C^{\perp})}{\phi(C^{\perp})}
\end{equation}
\end{definition}

The notion of Reichenbachian common causes can be generalized as follows.

\begin{definition}
\label{rccs-definition}
Let $\mathcal{L}$ be an orthomodular lattice, and let $\phi$ be a probability measure on $\mathcal{L}$. 

The set $\{ C_j \in \mathcal{L} | j \in J \}$ is called a partition in $\mathcal{L}$ if $\vee_{j \in J} C_j=1$, and $C_i$ and $C_j$ are orthogonal for $i \neq j$, where $J$ is an index set.

If $A$ and $B$ in $\mathcal{L}$ are correlated, then a partition $\{ C_j | j \in J \}$ is called a Reichenbachian common cause system for the correlation if $C_j$ is compatible with both $A$ and $B$ for every $j \in J$, and the following conditions hold:
\begin{equation}
\label{rccs1}
\frac{\phi(A \wedge B \wedge C_j)}{\phi(C_j)}=\frac{\phi(A \wedge C_j)}{\phi(C_j)}\frac{\phi(B \wedge C_j)}{\phi(C_j)} 
\end{equation}
for any $j \in J$, and
\begin{equation}
\label{rccs2}
\left( \frac{\phi(A \wedge C_i)}{\phi(C_i)} - \frac{\phi(A \wedge C_j)}{\phi(C_j)} \right) \left(\frac{\phi(B \wedge C_i)}{\phi(C_i)} - \frac{\phi(B \wedge C_j)}{\phi(C_j)} \right) > 0 
\end{equation}
for any mutually distinct elements $i, j \in J$.
The cardinality of the index set $J$ is called the size of the common cause system.
\end{definition}

The following Lemma shows that $A$ and $B$ are correlated in $\phi$ if Equations (\ref{rccs1}) and (\ref{rccs2}) hold.

\begin{lemma}
\label{correlation-rccs-lemma}
\citep{hofer2004reichenbachian, hofer2006reichenbachian}
Let $\mathcal{L}$ be an orthomodular lattice, let $\phi$ be a probability measure on $\mathcal{L}$, let $\{ C_j | j \in J \}$ be a partition of $\mathcal{L}$ such that $\phi(C_j)>0$ for any $j \in J$, and let $A$ and $B$ be mutually commuting elements in $\mathcal{L}$. If
\[ \frac{\phi(A \wedge B \wedge C_j)}{\phi(C_j)}=\frac{\phi(A \wedge C_j)}{\phi(C_j)}\frac{\phi(B \wedge C_j)}{\phi(C_j)},  \]
for any $j \in J$, then
\[ \phi(A \wedge B) - \phi(A)\phi(B) = \frac{1}{2} \sum_{i \neq j} \phi(C_i)\phi(C_j) \left( \frac{\phi(A \wedge C_i)}{\phi(C_i)} - \frac{\phi(A \wedge C_j)}{\phi(C_j)} \right)\left( \frac{\phi(B \wedge C_i)}{\phi(C_i)} - \frac{\phi(B \wedge C_j)}{\phi(C_j)} \right) \]
\end{lemma}

Let $A$ and $B$ be compatible elements in an orthomodular lattice $\mathcal{L}$, and let $\mathcal{B}$ be the Boolean sub-lattice of $\mathcal{L}$ which contains $A$ and $B$. If $A \wedge B^{\perp} = 0$, then $A=(A \wedge B) \vee (A \wedge (A \wedge B)^{\perp})=(A \wedge B) \vee (A \wedge B^{\perp})=A \wedge B$. It means that, for any truth-value assignment $h$ of $\mathcal{B}$, $h(A)=1$ entails $h(B)=1$. If the truth-value of $B$ is independent of that of $A$, $A \wedge B^{\perp}$ should not be $0$. This is a motivation of logical independence

\begin{definition}
\citep{redei1995logical, redei1995logically, redei1998quantum} 
Let $\mathcal{L}$ be an orthomodular lattice, and let $A$ and $B$ be elements in $\mathcal{L}$. If $A \wedge B \neq 0$, $A^{\perp} \wedge B^{\perp} \neq 0$, $A \wedge B^{\perp} \neq 0$, and $A^{\perp} \wedge B \neq 0$, then it is said that $A$ and $B$ are logical independent.
\end{definition}

Logical independence has the following property. This lemma is used in the proof of Theorem \ref{size-3}.

\begin{lemma}
\label{logical_independence_property}
Let $\mathcal{L}$ be an orthomodular lattice, and let $A$ and $B$ be compatible elements in $\mathcal{L}$. Then the following conditions are equivalent.
\begin{enumerate}
\item $A$ and $B$ are logical independent.
\item $A \vee B > A$, $A \vee B > B$, $A \vee B^{\perp} >A$, $A \vee B^{\perp} > B^{\perp}$
\end{enumerate}
\end{lemma}

\begin{proof}
\begin{description}
\item[$1 \Longrightarrow 2$]

Let $A$ and $B$ are compatible and logical independent elements.
Suppose $A=A \vee (A^{\perp} \wedge B)$. Then $A \geq A^{\perp} \wedge B$. Since $A$ and $B$ are compatible, 
\begin{equation}
B=(B \wedge A) \vee (B \wedge A^{\perp}) \leq (B \wedge A) \vee A = A.
\end{equation}
Then $A^{\perp} \wedge B \leq A^{\perp} \wedge A = 0$. It contradicts with logical independence between $A$ and $B$. 

The set of $A$, $A^{\perp}$, and $B$ are distributive because $A$ and $B$ are compatible \cite[p.25]{kalmbach1983orthomodular}. Thus 
\begin{equation}
\label{t9.9}
A < A \vee (A^{\perp} \wedge B) =A \vee B
\end{equation}
Similarly
\begin{equation}
\label{lip2}
B <   B \vee (A \wedge B^{\perp}) = A \vee B ,
\end{equation}
\begin{equation}
\label{lip3}
A <  A \vee (A^{\perp} \wedge B^{\perp}) = A \vee B^{\perp},
\end{equation}
\begin{equation}
\label{lip4}
B^{\perp} <  (A \wedge B) \vee B^{\perp} =A \vee B^{\perp}.
\end{equation}

\item[$2 \Longrightarrow 1$]

Let $A$ and $B$ be compatible elements in $\mathcal{L}$ which satisfies Condition 2. Then
\begin{equation}
\label{lip5}
A < A \vee B = A \vee (A^{\perp} \wedge B), 
\end{equation}
\begin{equation}
\label{lip6}
B < A \vee B = B \vee (A \wedge B^{\perp}),
\end{equation}
\begin{equation}
\label{lip7}
A < A \vee B^{\perp} = A \vee (A^{\perp} \wedge B^{\perp}),
\end{equation}
\begin{equation}
\label{lip8}
B^{\perp} < A \vee B^{\perp}  = B^{\perp} \vee  (A \wedge B).
\end{equation}
Therefore, $A$ and $B$ are logical independent.
\end{description}
\end{proof}

The following theorem shows that logical independence plays an essential role in a Reichenbachian common cause system. 

\begin{theorem}
\label{logical-independence-rccs-proposition}
Let $\mathcal{L}$ be an orthomodular lattice, let $\phi$ be a probability measure on $\mathcal{L}$, let $A$ and $B$ be compatible elements in $\mathcal{L}$ such that $\phi(A \wedge B) > \phi(A)\phi(B)$, and let $n$ be a natural number such that $n \geq 3$. If there exists a Reichenbachian common cause system of size $n$ for the correlation between $A$ and $B$, then $A$ and $B$ are logical independent.
\end{theorem}

\begin{proof}
Let $A$ and $B$ be compatible elements in $\mathcal{L}$ which are not logical independent and correlated in a probability measure $\phi$ on $\mathcal{L}$. By Lemma \ref{correlation-lemma}, $\phi(A \wedge B) > 0$ and $\phi(A^{\perp} \wedge B^{\perp}) > 0$. Thus 
\begin{equation}
\label{lemma-3}
A \wedge B \neq 0, \ \ \ \ A^{\perp} \wedge B^{\perp} \neq 0.
\end{equation}
Since $A$ and $B$ are not logical independent, $A \wedge B^{\perp}=0$ or $A^{\perp} \wedge B = 0$ by Equation (\ref{lemma-3}). By symmetry we do not need to differentiate between $A^{\perp} \wedge B=0$ and $A \wedge B^{\perp} = 0$, so we assume $A \wedge B^{\perp} = 0$. 

Since $A$ and $B$ are compatible,
\begin{equation}
\label{lemma-4}
A=(A \wedge B) \vee (A \wedge B^{\perp})=A \wedge B.
\end{equation}
Suppose that a Reichanbachian common cause system $\{ C_1, C_2, \dots, C_n \}$ for the correlation between $A$ and $B$ exists, where $n \geq 3$. By Equation (\ref{lemma-4})
\begin{equation}
\frac{\phi(A \wedge C_i)}{\phi(C_i)}\frac{\phi(B \wedge C_i)}{\phi(C_i)}=\frac{\phi(A \wedge B \wedge C_i)}{\phi(C_i)}=\frac{\phi(A \wedge C_i)}{\phi(C_i)}.
\end{equation}
Thus, for any $i \in \{ 1, 2, \dots ,n \}$, either $\phi(A \wedge C_i)/\phi(C_i)=0$ or $\phi(B \wedge C_i)/\phi(C_i)=1$ holds. Therefore there are mutually distinct natural numbers $i,j \in \{ 1, 2, \dots ,n \}$ such that either
\begin{equation}
\frac{\phi(A \wedge C_i)}{\phi(C_i)}=\frac{\phi(A \wedge C_j)}{\phi(C_j)}=0 \ \ \text{or} \ \  \frac{\phi(B \wedge C_i)}{\phi(C_i)}=\frac{\phi(B \wedge C_j)}{\phi(C_j)}=1
\end{equation}
because $n \geq 3$. It implies
\begin{equation}
\left( \frac{\phi(A \wedge C_i)}{\phi(C_i)} - \frac{\phi(A \wedge C_j)}{\phi(C_j)} \right) \left(\frac{\phi(B \wedge C_i)}{\phi(C_i)} - \frac{\phi(B \wedge C_j)}{\phi(C_j)} \right) = 0.
\end{equation}
Thus $\{ C_1, C_2, \dots, C_n \}$ does not satisfy Equation (\ref{rccs2}) in Definition \ref{rccs-definition}. Therefore no common cause system for the correlation between $A$ and $B$ exists.
\end{proof}

According Theorem \ref{logical-independence-rccs-proposition}, if correlated events are not logical independent, there is no Reichenbachian common cause system for this correlation. Therefore logical independence is essential for Reichenbachian common cause systems.

A notion of causally $n$-closedness is obtained if the notion of common causes is replaced with the concept of common cause systems.

\begin{definition}
Let $\mathcal{L}$ be an orthomodular lattice, and let $\phi$ be a probability measure on $\mathcal{L}$. The probability space $(\mathcal{L}, \phi)$ is called causally $n$-closed with respect to logical independence if for any correlation between compatible and logical independent elements there exists a Reichenbachian common cause system of size $n$ in $(\mathcal{L}, \phi)$.
\end{definition}

In Theorem \ref{size-3}, we show that a general probability theory which is represented by an atomless orthomodular lattice contains Reichenbachian common cause systems of size 3 for any correlation between two logical independent events.

\begin{theorem}
\label{size-3}
Let $\mathcal{L}$ be a $\sigma$-complete atomless orthomodular lattice, and $\phi$ a faithful $\sigma$-additive probability measure on $\mathcal{L}$. Then the probability space $(\mathcal{L}, \phi)$ is causally $3$-closed with respect to logical independence.
\end{theorem}

\begin{proof}
Let $A$ and $B$ be mutually compatible and logical independent elements in $\mathcal{L}$ such that 
\begin{equation}
\label{t1}
\phi(A \wedge B)>\phi(A)\phi(B). 
\end{equation}
By Lemma \ref{inequality_lemma}
\begin{equation}
\label{t2}
\phi(A)\phi(B) \geq \phi(A \wedge B)\phi(A \vee B)
\end{equation}
Equations (\ref{t1}) and (\ref{t2}) imply 
\begin{equation}
\label{t2.2}
0 < \phi(A \wedge B) - \phi(A)\phi(B) \leq \phi(A \wedge B)(1-\phi(A \vee B)),
\end{equation}
and 
\begin{equation}
\label{t2.01}
\phi(A \vee B) < 1
\end{equation}
by Lemma \ref{correlation-lemma}.
Inequalities (\ref{t2.2}) and (\ref{t2.01}) entail
\begin{equation}
\label{t3}
0 < \frac{\phi(A \wedge B) - \phi(A)\phi(B)}{1-\phi(A \vee B)} \leq \phi(A \wedge B). 
\end{equation}
By Lemma \ref{denseness_lemma} and Inequality (\ref{t3}) there exists an element $C_1$ such that 
\begin{equation}
\label{t3.9}
0 < C_1 < A \wedge B
\end{equation}
and
\begin{equation}
\label{t4}
0 < \phi(C_1) < \frac{\phi(A \wedge B) - \phi(A)\phi(B)}{1-\phi(A \vee B)} \leq \phi(A \wedge B). 
\end{equation}
Since $C_1$ is compatible with $A \wedge B$, $A$, and $B$,
\begin{equation}
\label{t4.1}
\phi(X \wedge C_1^{\perp})=\phi(X) - \phi(X \wedge C_1) = \phi(X)-\phi(C_1) > 0,
\end{equation}
where $X=A \wedge B, A, B$.
By Inequality (\ref{t3.9})
\begin{equation}
\label{c1-property}
\frac{\phi(A \wedge C_1)}{\phi(C_1)}=\frac{\phi(B \wedge C_1)}{\phi(C_1)}=\frac{\phi(A \wedge B \wedge C_1)}{\phi(C_1)}=1.
\end{equation}
Inequality (\ref{t4}) and Equation (\ref{t4.1}) imply
\begin{equation}
\label{t5.8}
\begin{split}
&\phi(C_1^{\perp})\phi(A \wedge B \wedge C_1^{\perp}) - \phi(A \wedge C_1^{\perp})\phi(B \wedge C_1^{\perp}) \\
&= (1-\phi(C_1))(\phi(A \wedge B) - \phi(C_1)) - (\phi(A) - \phi(C_1))(\phi(B)-\phi(C_1)) \ \ \ \ \ (\because \text{Equation (\ref{t4.1})}) \\
&= \phi(A \wedge B) - \phi(A)\phi(B) - \phi(C_1)(1-\phi(A)-\phi(B)+\phi(A \wedge B)) \\
&=\phi(A \wedge B) - \phi(A)\phi(B) - \phi(C_1)(1 - \phi(A \vee B)) \\
&>0 \ \ \ \ (\because \text{Inequality (\ref{t4}}))
\end{split}
\end{equation}
Since $C_1$ is compatible with both $A$ and $B$,
\begin{equation}
\label{t5}
\phi(A \wedge C_1^{\perp} ) \phi(B \wedge C_1^{\perp}) \geq \phi((A \wedge B) \wedge C_1^{\perp})\phi((A \vee B) \wedge C_1^{\perp}),
\end{equation}
by Lemma \ref{inequality_lemma}.

By Inequality (\ref{t3.9}) $C_1^{\perp} > (A \wedge B)^{\perp} \geq (A \vee B)^{\perp}$. Thus $C_1^{\perp}$ and $(A \vee B)^{\perp}$ are compatible. Therefore
\begin{equation}
\label{t5.7}
\phi(C_1^{\perp})=\phi(C_1^{\perp} \wedge (A \vee B)^{\perp})+\phi(C_1^{\perp} \wedge (A \vee B))=\phi((A \vee B)^{\perp})+\phi((A \vee B) \wedge C_1^{\perp}).
\end{equation}
Inequalities (\ref{t4.1}) and (\ref{t5.8}) entail
\begin{equation}
\label{t5.9}
\begin{split}
0 &<\phi(C_1^{\perp})- \frac{\phi(A \wedge C_1^{\perp})\phi(B \wedge C_1^{\perp})}{\phi(A \wedge B \wedge C_1^{\perp})} \\
&\leq \phi(C_1^{\perp}) - \phi((A \vee B) \wedge C_1^{\perp})  \ \ \ \ (\because \text{Inequality (\ref{t5}))} \\
&= \phi((A \vee B)^{\perp})  \ \ \ \ (\because \text{Equation (\ref{t5.7}))} \\
&=\phi(A^{\perp} \wedge B^{\perp}).
\end{split}
\end{equation}
By Lemma \ref{denseness_lemma} there exists an element $C_2 \in \mathcal{L}$ such that 
\begin{equation}
\label{t6}
C_2 \leq A^{\perp} \wedge B^{\perp} 
\end{equation} 
and
\begin{equation}
\label{t7}
\phi(C_2)=\phi(C_1^{\perp}) - \frac{\phi(A \wedge C_1^{\perp})\phi(B \wedge C_1^{\perp})}{\phi(A \wedge B \wedge C_1^{\perp})} > 0.
\end{equation}
By Inequality (\ref{t6})
\begin{equation}
\label{c2-property}
\frac{\phi(A \wedge C_2)}{\phi(C_2)}=\frac{\phi(B \wedge C_2)}{\phi(C_2)}=\frac{\phi(A \wedge B \wedge C_2)}{\phi(C_2)}=0.
\end{equation}
Let $C_3 := C_1^{\perp} \wedge C_2^{\perp}$. Then $C_1$, $C_2$, and $C_3$ are mutually orthogonal elements such that $C_1 \vee C_2 \vee C_3 = 1$. Thus $\{ C_1, C_2, C_3 \}$ is a partition in $\mathcal{L}$. Since $(A \wedge B) \vee (A \wedge B^{\perp}) \vee (A^{\perp} \wedge B) \vee (A^{\perp} \wedge B^{\perp})  = 1$ by \citet[p.26]{kalmbach1983orthomodular},
\begin{equation}
\label{t0.9}
\phi(A \wedge B)+\phi(A \wedge B^{\perp})+\phi(A^{\perp} \wedge B)+\phi(A^{\perp} \wedge B^{\perp}) =1.
\end{equation} 
Thus
\begin{equation}
\begin{split}
\phi(C_3) &= 1 - \phi(C_1) - \phi(C_2) \\
&> 1 - \phi(A \wedge B) - \phi(A^{\perp} \wedge B^{\perp}) \ \ \ (\because \text{Equation (\ref{t7}), and Inequalities (\ref{t4}) and (\ref{t5.9})}) \\
&= \phi(A \wedge B^{\perp})+\phi(A^{\perp} \wedge B) \ \ \ (\because \text{Equation (\ref{t0.9})}) \\
&\geq 0.
\end{split}
\end{equation}
By Equation (\ref{t7})
\begin{equation}
\frac{\phi(A \wedge C_3)}{\phi(C_3)}\frac{\phi(B \wedge C_3)}{\phi(C_3)}=\frac{\phi(A \wedge B \wedge C_3)}{\phi(C_3)}
\end{equation}
since
\begin{equation}
\phi(C_3)=1 - \phi(C_1) - \phi(C_2)=\phi(C_1^{\perp}) - \phi(C_2)
\end{equation}
and
\begin{equation}
\label{t8}
X \wedge C_3 = X \wedge C_1^{\perp} \wedge C_2^{\perp} = X \wedge C_1^{\perp}
\end{equation}
where $X=A, B, A \wedge B$.

By Lemma \ref{logical_independence_property}
\begin{equation}
\label{t10}
\phi(A \vee B) > \phi(A), \ \ \ \ \phi(A \vee B) > \phi(B)
\end{equation}
since $\phi$ is faithful, and $A$ and $B$ are compatible and logical independent.

Equation (\ref{t8}) and Inequality (\ref{t4}) imply
\begin{equation}
\begin{split}
&\phi(A \wedge C_3)=\phi(A \wedge C_1^{\perp})=\phi(A)-\phi(C_1) > 0 \\
&\phi(B \wedge C_3)=\phi(B \wedge C_1^{\perp})=\phi(B)-\phi(C_1) > 0,
\end{split}
\end{equation}
and Inequalities (\ref{t5.9}) and (\ref{t10}), and Equation (\ref{t7}) entail
\begin{equation}
\begin{split}
&\phi(C_2) \leq \phi(A^{\perp} \wedge B^{\perp}) = 1 - \phi(A \vee B) < 1 - \phi(A) \\
&\phi(C_2) \leq \phi(A^{\perp} \wedge B^{\perp}) = 1 - \phi(A \vee B) < 1 - \phi(B).
\end{split}
\end{equation}
Thus
\begin{equation}
\label{c3-property}
\begin{split}
&0< \frac{\phi(A \wedge C_3)}{\phi(C_3)}=\frac{\phi(A) - \phi(C_1)}{1- \phi(C_2) - \phi(C_1)} < \frac{1- \phi(C_2) - \phi(C_1)}{1- \phi(C_2) - \phi(C_1)} = 1 \\ 
&0< \frac{\phi(B \wedge C_3)}{\phi(C_3)} =\frac{\phi(B) - \phi(C_1)}{1- \phi(C_2) - \phi(C_1)} < \frac{1- \phi(C_2) - \phi(C_1)}{1- \phi(C_2) - \phi(C_1)} = 1.
\end{split}
\end{equation}
By Equations (\ref{c1-property}) and (\ref{c2-property}), and Inequality (\ref{c3-property})
\begin{equation}
\left( \frac{\phi(A \wedge C_i)}{\phi(C_i)} - \frac{\phi(A \wedge C_j)}{\phi(C_j)} \right) \left(\frac{\phi(B \wedge C_i)}{\phi(C_i)} - \frac{\phi(B \wedge C_j)}{\phi(C_j)} \right) > 0 \ \ \ (i \neq j) .
\end{equation}
\end{proof}

By Proposition \ref{logical-independence-rccs-proposition} and Theorem \ref{size-3}, we get the following corollary.

\begin{corollary}
\label{corollary_size-3}
Let $\mathcal{L}$ be a $\sigma$-complete atomless orthomodular lattice, let $\phi$ be a faithful $\sigma$-additive probability measure on $\mathcal{L}$, and let $A$ and $B$ be compatible elements in $\mathcal{L}$ such that $\phi(A \wedge B) > \phi(A)\phi(B)$. There exists a Reichenbachian common cause system of size 3 for the correlation between $A$ and $B$ if and only if $A$ and $B$ are logical independent. 
\end{corollary}

This corollary shows that the existence of a Reichenbachian common cause system of size 3 for the correlation between $A$ and $B$ is equivalent to logical independence of $A$ and $B$ in the case of a general probability theory which is represented by a $\sigma$-complete atomless orthomodular lattice with a faithful $\sigma$-additive probability measure.

Theorem \ref{size-3} and Corollary \ref{corollary_size-3} can apply to AQFT because typical local algebras in algebraic quantum field theory are type III von Neumann algebras and lattices of projections of type II or type III von Neumann algebras is atomless orthomodular lattices. In algebraic quantum field theory, each bounded open region $\mathcal{O}$ in the Minkowski space is associated with a von Neumann algebra $\mathfrak{N}(\mathcal{O})$ on a Hilbert space $\mathcal{H}$. Such a von Neumann algebra is called a local algebra. We say that bounded open regions $\mathcal{O}_1$ and $\mathcal{O}_2$ are strictly space-like separated if there is a neighborhood $\mathcal{V}$ of the origin of the Minkowski space such that $\mathcal{O}_1 + \mathcal{V}$ and $\mathcal{O}_2$ are space-like separated.

If $\mathcal{O}_1$ and $\mathcal{O}_2$ are strictly space-like separated bounded open regions in the Minkowski space, local algebras $\mathfrak{N}(\mathcal{O}_1)$ and $\mathfrak{N}(\mathcal{O}_2)$ are logical independent under usual axioms \cite[Theorem 1.12.3]{baumgartel1995operatoralgebraic}. Thus there exists a Reichenbachian common cause system of size 3 for any correlation between $\mathfrak{N}(\mathcal{O}_1)$ and $\mathfrak{N}(\mathcal{O}_2)$ by Corollary \ref{corollary_size-3}. 

On the other hand, it is known that there is a normal state on $\mathfrak{N}(\mathcal{O}_1) \vee \mathfrak{N}(\mathcal{O}_2)$ which violates Clauser-Horne-Shimony-Holt Bell inequality \citep{halvorson2000generic, kitajima2013epr, landau1987violation, summers1987bell, summers1987abell, summers1987maximal, summers1988maximal}. Clauser-Horne Bell inequality is also violated in some normal state on $\mathfrak{N}(\mathcal{O}_1) \vee \mathfrak{N}(\mathcal{O}_2)$\footnote{\cite{stergiou2012two} discussed a relation between Clauser-Horne-Shimony-Holt Bell inequality and Clauser-Horne Bell inequality in AQFT.}. Since $\mathfrak{N}(\mathcal{O}_1)$ and $\mathfrak{N}(\mathcal{O}_2)$ are not abelian, there are partial isometries $V_{1} \in \mathfrak{N}(\mathcal{O}_1)$ and $V_{2} \in \mathfrak{N}(\mathcal{O}_2)$ such that $V_1^2=V_2^2=0$ \citep{landau1987violation}. Since $\mathfrak{N}_1$ and $\mathfrak{N}_2$ are logical independent, $(V_1^*V_1)(V_2^*V_2) \neq 0$. Thus there is a unit vector $\Psi$ such that $(V_1^*V_1)(V_2^*V_2)\Psi \neq 0$. Let
\begin{equation}
 \Phi := \frac{1}{\sqrt{2}} (\Psi + V_1V_2 \Psi), 
\end{equation}
\begin{equation}
\begin{split}
&A_1 := V_1^*V_1, \ \ \ \ \ B_1:=\frac{3}{4}V_1^*V_1+\frac{1}{4}V_1V_1^*+\frac{\sqrt{3}}{4}(V_1+V_1^*), \\
&A_2 := V_2^*V_2, \ \ \ \ \ B_2:=\frac{3}{4}V_2^*V_2+\frac{1}{4}V_2V_2^* - \frac{\sqrt{3}}{4}(V_2+V_2^*).
\end{split}
\end{equation}
$A_1$, $B_1$, $A_2$, and $B_2$ are projections.
Let $\phi$ be a vector state of $\mathfrak{N}(\mathcal{O}_1) \vee \mathfrak{N}(\mathcal{O}_2)$ induced by $\Phi$. Then
\begin{equation}
\label{bell_i}
\phi(A_{1})+\phi(A_{2})+\phi(B_{1}B_{2})-\phi(A_{1}A_{2})-\phi(B_{1}A_{2})-\phi(A_{1}B_{2}) = -\frac{1}{8} < 0.
\end{equation}

Suppose that there is a Reichenbachian `common' common cause system $\{ C_1, \dots, C_{n} \}$ for all four pairs $(A_1, A_2)$, $(A_1, B_2)$, $(B_1, A_2)$, and $(B_1, B_2)$ \citep{redei1997reichenbach, hofer1999reichenbach, redei2002local}. In other words,
\begin{equation}
\label{rccy}
\begin{split}
&\frac{\phi(A_1 A_2 C_j)}{\phi(C_j)}=\frac{\phi(A_1 C_j)}{\phi(C_j)}\frac{\phi(A_2 C_j)}{\phi(C_j)}, \\
&\frac{\phi(A_1 B_2 C_j)}{\phi(C_j)}=\frac{\phi(A_1 C_j)}{\phi(C_j)}\frac{\phi(B_2 C_j)}{\phi(C_j)}, \\
&\frac{\phi(B_1 A_2 C_j)}{\phi(C_j)}=\frac{\phi(B_1 C_j)}{\phi(C_j)}\frac{\phi(A_2 C_j)}{\phi(C_j)}, \\
&\frac{\phi(B_1 B_2 C_j)}{\phi(C_j)}=\frac{\phi(B_1 C_j)}{\phi(C_j)}\frac{\phi(B_2 C_j)}{\phi(C_j)}, \\
&[A_1,C_j]=[A_2,C_j]=[B_1,C_j]=[B_2,C_j]=0, \\
&\sum_{i=1}^n C_i = I,
\end{split}
\end{equation}
for any $j \in \{1, \dots, n \}$.

Then
\begin{equation}
0 \leq \phi(A_{1})+\phi(A_{2})+\phi(B_{1}B_{2})-\phi(A_{1}A_{2})-\phi(B_{1}A_{2})-\phi(A_{1}B_{2}) \leq 1
\end{equation}
since
\begin{equation}
\begin{split}
&\phi(A_{1})+\phi(A_{2})+\phi(B_{1}B_{2})-\phi(A_{1}A_{2})-\phi(B_{1}A_{2})-\phi(A_{1}B_{2}) \\
&=\phi \left( \left( A_{1}+A_{2}+B_{1}B_{2}-A_{1}A_{2}-B_{1}A_{2}-A_{1}B_{2} \right) \left( \sum_{i=1}^{n}C_{i} \right) \right) \\
&=\sum_{i=1}^{n} \left( \frac{\phi(A_{1}C_{i})}{\phi(C_{i})} +\frac{\phi(A_{2}C_{i})}{\phi(C_{i})} +\frac{\phi(B_{1}B_{2}C_{i})}{\phi(C_{i})} -\frac{\phi(A_{1}A_{2}C_{i})}{\phi(C_{i})}-\frac{\phi(B_{1}A_{2}C_{i})}{\phi(C_{i})}-\frac{\phi(A_{1}B_{2}C_{i})}{\phi(C_{i})}
\right)\phi(C_{i}) \\
&=\sum_{i=1}^{n} \huge( \frac{\phi(A_{1}C_{i})}{\phi(C_{i})} +\frac{\phi(A_{2}C_{i})}{\phi(C_{i})} +\frac{\phi(B_{1}C_{1})}{\phi({C_{i}})}\frac{\phi(B_{2}C_{i})}{\phi(C_{i})} -\frac{\phi(A_{1}C_{1})}{\phi({C_{i}})}\frac{\phi(A_{2}C_{i})}{\phi(C_{i})} \\
& \ \ \ \ \ \ \ \ \ \ \ \ \ -\frac{\phi(B_{1}C_{1})}{\phi({C_{i}})}\frac{\phi(A_{2}C_{i})}{\phi(C_{i})}-\frac{\phi(A_{1}C_{1})}{\phi({C_{i}})}\frac{\phi(B_{2}C_{i})}{\phi(C_{i})}
\huge)\phi(C_{i}) \\
\end{split}
\end{equation}
by Equations (\ref{rccy})
and 
\begin{equation}
\begin{split}
 &a_{1}+a_{2}+b_{1}b_{2}-a_{1}a_{2}-b_{1}a_{2}-a_{1}b_{2} \\
 &=a_{1}(b_{1}(1-a_{2})+(1-b_{1})(1-b_{2}))+(1-a_{1})(b_{1}b_{2}+(1-b_{1})a_{2})
 \end{split}
 \end{equation}
for any $a_{1},a_{2},b_{1},b_{2} \in [0,1]$.
It contradicts Inequality (\ref{bell_i}). Thus, there is no Reichenbachian `common' common cause system $\{ C_{1}, \dots, C_{n} \}$ for all four pairs $(A_1, A_2)$, $(A_1, B_2)$, $(B_1, A_2)$, and $(B_1, B_2)$ because Bell's inequality does not hold in $\phi$. It means that the existence of Reichenbachian `common' cause system contradicts the violation of Bell's inequality \citep[Section 9.2]{hofer2013principle}.

On the other hand, the definition of Reichenbachian common cause systems does not require that it is common for all four pairs (Definition \ref{rccs-definition}). Thus, the existence of a Reichenbachian common cause system can be compatible with the violation of Bell's inequality in AQFT, and Theorem \ref{size-3} and Corollary \ref{corollary_size-3} show that there are Reichenbachian common cause systems of size 3 for any correlation between $\mathfrak{N}(\mathcal{O}_1)$ and $\mathfrak{N}(\mathcal{O}_2)$ in spite of the violation of Bell's inequality.

\section{Conclusion}

In the present paper, we examined Reichenbachian common cause systems in general probability theories which include classical probability theories and quantum probability theories. Especially, the following two problems were investigated:
\begin{itemize}
\item Is logical independence necessary for the existence of a Reichenbachian common cause system?
\item Is there a general probability theory which is causal $n$-closed for some natural number $n$?
\end{itemize}

In Theorem \ref{logical-independence-rccs-proposition}, it was shown that there is no Reichenbachian common cause system for a correlation between two events if they are not logical independent. Therefore logical independence is necessary for the existence of a Reichenbachian common cause system. In Theorem \ref{size-3} it was shown that a general probability theory which is represented by an atomless orthomodular lattice with a faithful $\sigma$-additive probability measure is causally $3$-closed. It remains open, however, whether this lattice is causally $n$-closed for any natural number $n$.

Theorem \ref{size-3} and Corollary \ref{corollary_size-3} can apply to AQFT because typical local algebras in algebraic quantum field theory are type III von Neumann algebras and lattices of projections of type II or type III von Neumann algebras is atomless orthomodular lattices. There exists a Reichenbachian common cause system of size 3 for any correlation between two space-like separated regions by Corollary \ref{corollary_size-3}. On the other hand, AQFT predicts many states which violate Bell's inequality between two events associated with space-like separated spacetime regions in Minkowski spacetime \citep{halvorson2000generic, kitajima2013epr, landau1987violation, summers1987bell, summers1987abell, summers1987maximal, summers1988maximal}. After Corollary \ref{corollary_size-3}, we pointed out that this violation of Bell's inequality can be compatible with a Reichenbachian common cause system although it contradicts a `common' common cause system.

\cite{redei2002local, redei2007remarks} showed that a Reichenbachian common cause for any correlation exists in the union of the backward light cones of $\mathcal{O}_1$ and $\mathcal{O}_2$ in AQFT. By this result and Theorem \ref{size-3}, it can be shown that a Reichenbachian common cause system of size 3 for any correlation exists in the union of the backward light cones of $\mathcal{O}_1$ and $\mathcal{O}_2$. However, the problem about Reichenbachian commmon cause systems in AQFT is largely open.

\section*{Acknowledgments}
The authors thank Chrysovalantis Stergiou for a helpful comment for an earlier version of this paper. The author is supported by the JSPS KAKENHI No.15K01123 and No.23701009.

\bibliographystyle{chicago}
\bibliography{kitajima}

\end{document}